\documentclass[varenna,graybox,vecphy]{cimento}
\usepackage{helvet}         
\usepackage{courier}
\usepackage{amsthm}
\usepackage{bbm}
\usepackage{changes}

\newtheorem{thm}{Theorem}
\newtheorem{thmm}{Theorem}
\newtheorem{defnn}[thmm]{Definition}

\newtheorem{lemma}{Lemma}
\newtheorem*{lemma*}{Lemma}
\theoremstyle{definition}

\usepackage{makeidx}         
\usepackage{graphicx}        
\usepackage{multicol}        
\usepackage[bottom]{footmisc}

\ifcsname SQRT \endcsname
\else
 
\fi

\newcommand{\ensm}[1]{\ensuremath{#1}}

\renewcommand{\alph}{\ensm{\alpha}}

\newcommand{\bra}[1]{\ensm{\langle #1|}}

\newcommand{\ket}[1]{\ensm{| #1 \rangle}}

\newcommand{\be}{\begin{equation}}
\newcommand{\ee}{\end{equation}}
\newcommand{\bea}{\begin{eqnarray}}
\newcommand{\eea}{\end{eqnarray}}
\newcommand{\proj}[1]{\ket{#1}\!\bra{#1}}
\newcommand{\ot}[0]{\otimes}
\newcommand{\Tr}[1]{\mathrm{Tr}#1}
\newcommand{\tr}[1]{\mathrm{tr}#1}

\title{Entanglement and Nonlocality in Many-Body Systems: a primer}
\author{J. Tura, A.B. Sainz, T. Gra\ss,
R. Augusiak}
\institute{ICFO -- Institut de Ci\`encies Fot\`oniques, Av. C.F. Gauss 3, Castelldefels, 08860 Spain}
\author{A. Ac\'in \atque M.~Lewenstein}
\institute{ICFO -- Institut de Ci\`encies Fot\`oniques, Av. C.F. Gauss 3, Castelldefels, 08860 Spain}
\institute{ICREA -- Instituci\'o Catalana de Recerca i Estudis Avan\c cats, Lluis Campanys 3,  Barcelona, 08010 Spain}

\PACSes{\PACSit{03.65.Ud, 03.67.Bg, 03.67.Mn}{\ldots}}

\begin{document}

\maketitle

\begin{abstract}
Current understanding of correlations and quantum phase
transitions in many-body systems has significantly improved thanks
to the recent intensive studies of their entanglement properties.
In contrast, much less is known about the role of quantum
non-locality in these systems. On the one hand, standard,
"theorist- and experimentalist-friendly" many-body observables
involve correlations among only few (one, two, rarely three...)
particles. On the other hand, most of the available multipartite
Bell inequalities involve correlations among many particles. Such
correlations are notoriously hard to access theoretically, and
even harder experimentally. Typically, there is no Bell inequality
for many-body systems built only from low-order correlation
functions. Recently, however, it has been shown in
[J. Tura \textit{et al.}, Science {\bf 344}, 1256 (2014)]
that multipartite Bell inequalities constructed only from two-body correlation functions are strong enough to reveal
non-locality in some many-body states, in particular those relevant for
nuclear and atomic physics. The purpose of this lecture is to provide an overview of the
problem of quantum correlations in many-body systems -- from
entanglement to nonlocality -- and the methods for their
characterization.
\end{abstract}

\section{Introduction}

Nonlocality is a property of correlations that goes beyond the
paradigm of local realism
\cite{EPR,Bohmbook,Bell64,Bellbook,NLreview,Bellissue}. According
to the celebrated theorem by J.S. Bell \cite{Bell64}, correlations
among the results of measurements of local observables performed
on some entangled states do not admit a {\it local
hidden-variable} (LHV) model (cf. \cite{RAMDAA} for a review on
LHV models). In other words, these correlations cannot be
described by observers who have access only to correlated
classical variables. In such instances, the observed quantum
correlations are named nonlocal and we talk about quantum
nonlocality, or Bell nonlocality. This can be detected by means of
the so-called Bell inequalities \cite{Bell64} - the celebrated
example of such is the famous Clauser-Horne-Shimony-Holt
inequality (CHSH) \cite{CHSH}. In general, Bell inequalities are
inequalities formulated in terms of linear combinations of the
probabilities observed when performing the local measurements on
composite systems, and their violation signals nonlocality.
Quantum, or Bell nonlocality is interesting for at least  three
reasons:
\begin{itemize}

\item It is a resource for quantum communication,  secure key
distribution~\cite{key1,key2,key3}, or certified quantum
randomness generation~\cite{rand1,rand2,rand3}. Hence, it is one
of the most important elements of the
 future quantum technologies.

\item It lies at the heart of philosophical aspects of quantum physics \cite{ZurekWheeler,Bellbook,Bohmbook},
leading frequently to controversial interpretations (see for
instance \cite{Heisenberg-free,FreeWillbook}, or the more recent
works \cite{Eberly1,Eberly2}).

\item Its characterization is a challenging complex and  difficult problem, proved to be, depending  on formulation,
 NP-complete or NP-hard
(\cite{Babai,Avis}; see also \cite{NLreview,Bellissue} and
references therein).
\end{itemize}

Quantum-mechanical states that violate Bell inequalities are
necessarily entangled and cannot be represented as mixtures of
projections on simple product states \cite{Werner1989} (for a
review on entanglement see \cite{Horodeccy}); the opposite does
not have to be true. Already in 1991 Gisin proved \cite{Gisin1991}
that any pure state of two parties violates a Bell's inequality.
This result was extended to an arbitrary number of parties by Popescu
and Rohlich \cite{PR}. But, Werner in the seminal paper from 1989
\cite{Werner1989} constructed examples of mixed bipartite states
that admit a LHV model for local projective measurements, and
nevertheless are entangled. This result was then generalized by
Barrett to arbitrary generalized measurements \cite{Barrett}. Very
recently, it has been shown that entanglement and nonlocality are
inequivalent for any number of parties (\cite{Demian2014} and
references therein).

On the other hand, entanglement, despite being a weaker property
of quantum states than nonlocality, has proven to be very useful
to characterize properties of many-body systems, and the nature of
quantum phase transitions (QPT) \cite{Sachdev}. For instance,
focusing on lattice spin models described by local
Hamiltonians, the following properties are true (for
a review see \cite{Augusiak2010,Lewenstein2012}):
\begin{itemize}

\item The reduced density matrix for two spins typically exhibits  entanglement
for  short separations of the spins only, even at criticality; still
entanglement measures show signatures of QPTs
\cite{Amico,Osborne2};

\item By performing optimized measurements on
the rest of the system, one can concentrate the entanglement in the
chosen two spins. One obtains in this way {\it localizable
entanglement} \cite{Frank1,Frank2}, whose entanglement length
diverges when the standard correlation length diverges, i.e., at
standard QPTs;

\item For non-critical systems, ground states (GSs) and low energy states exhibit the, so-called, area laws: the
von Neuman (or R\'enyi) entropy of the reduced density matrix of a
block of size $R$ scales as the size of the boundary of the block,
$\partial R$; at criticality logarithmic divergence occurs
frequently \cite{VidalLatorre} (for a review see
\cite{CardyJPA,EisertRMP}). These results are very well
established in 1D, while there are plenty of open questions in 2D
and higher dimensions;
\item GSs and low energy states can be efficiently described by the, so called, matrix product states, or more generally
tensor network states
(cf. \cite{Frank/Cirac2});

\item Topological order (at least for gapped systems in 1D and 2D) exhibits itself in the properties of the,
so called, {\it entanglement spectrum}, i.e. the spectrum of the
logarithm of the reduced density matrix of a block $R$
\cite{Haldane}, and in 2D in the appearance of the, so called,
{\it topological entropy}, i.e. negative constant correction to
the area laws \cite{Preskill,Wen}.
\end{itemize}

A natural question thus arises: Does non-locality play also an
important role in characterization of correlations in many-body
systems? Apart from its fundamental interest, so far the role of
nonlocality in such systems has hardly been explored. As already
mentioned, entanglement and nonlocality are known to be
inequivalent quantum resources. In principle, a generic many-body
state, say a ground state of a local Hamiltonian, is pure,
entangled and, because all pure entangled states violate a Bell
inequality~\cite{PR}, it is also nonlocal. However, this result
is hardly verifiable in experiments, because the known Bell
inequalities (see, e.g., \cite{nier1,nier2,Sliwa,JD1,nier4})
usually involve products of observables of all parties. Unfortunately,
measurements of such observables, although in principle possible
\cite{Greiner-single,Bloch-single}, are technically extremely
difficult; instead one has typically "easy" access to few-body
correlations, say one- and  two-body, in generic many-body
systems. Thus, the physically relevant question concerning the
nonlocality of many-body quantum states is whether its detection
is possible using only two-body correlations.

The plan of these lectures is the following: In Section \ref{Crash} we
present a crash course in entanglement theory, and talk about
bipartite pure and mixed states, about entanglement criteria, and
entanglement measures. Section \ref{ManyBody} is devoted to the
discussion of some aspects of entanglement in many-body systems.
There we talk about the computational complexity of many-body
problems, and relate it to entanglement of a generic state. We then
explain area laws, and indicate why they give us hopes
to find new efficient ways of solving many-body problems with new
numerical tools. These new tools are provided by the tensor
network states. Section \ref{Area} introduces the problem of
non-locality in many-body systems; we use here the contemporary
approach called {\it device-independent quantum information
theory} that talks about properties of correlations between
measurements only. Here we introduce the concept of classical
correlations, quantum-mechanical correlations, and non-signalling
correlations. CHSH inequality and its violations are shortly
presented here. In Section \ref{Nonlocality} we enter into the
problem of nonlocality detection in many-body systems based on
Bell inequalities that involve only two- and one-body correlators.
Here we explain the idea of permutationally invariant Bell
inequalities.  Finally, Section \ref{Conclusions} discusses
physical realizations of many-body non-locality  with ionic and
atomic models. These are promising systems in which the quantum
violation of our Bell inequalities could be observed. 


\section{Crash course on entanglement}
\label{Crash}

In this section, we focus on bipartite composite systems and
follow the presentation of Ref. \cite{Augusiak2010}. We will
define formally what entangled states are, and present one
important criterion to discriminate entangled states from
separable ones. However, before going into details, let us
introduce the notation. In what follows we will be mostly
concerned with bipartite scenarios, in which traditionally the
main roles are played by two parties called Alice and Bob. Let
\(\mathcal{H}_A\) denote the Hilbert space of Alice's physical
system, and \(\mathcal{H}_B\) that of Bob's. Our considerations
will be restricted to finite-dimensional Hilbert spaces, so we can
set $\mathcal{H}_{A}=\mathbbm{C}^{m}$ and
$\mathcal{H}_{B}=\mathbbm{C}^{n}$. Thus, the joint physical system
of Alice and Bob is described by the tensor product Hilbert space
\(\mathcal{H}_{AB}=\mathcal{H}_A \otimes
\mathcal{H}_{B}=\mathbbm{C}^{m}\ot\mathbbm{C}^{n}\). Finally,
$\mathcal{B}(\mathcal{H})$ will denote the set of bounded linear
operators from the Hilbert space $\mathcal{H}$ to $\mathcal{H}$.

\subsection{Bipartite pure states: Schmidt decomposition}
\label{subsec:2}

We begin our considerations from pure states for which it is much easier to
introduce the concept of entanglement.


\begin{defnn}\label{SepBipPureStates}
We call a pure state $|\psi_{AB}\rangle\in\mathcal{H}_{AB}$ \textit{separable} if there exist
pure states $|\psi_{A}\rangle\in\mathcal{H}_{A}$ and
$|\psi_{B}\rangle\in\mathcal{H}_{B}$ such that
$|\psi_{AB}\rangle=|\psi_{A}\rangle\otimes |\psi_{B}\rangle$.
Otherwise, $\ket{\psi_{AB}}$ is called \textit{entangled}.
\end{defnn}


To give an illustrative example of an entangled state from $\mathcal{H}_{AB}$
let us consider the
%
%
\textit{maximally entangled states}:
\begin{equation}\label{MaximallyEntState}
\ket{\psi_{+}^{(d)}}=\frac{1}{\sqrt{d}}\sum_{i=0}^{d-1}\ket{i}_{A}\ot \ket{i}_{B},
\end{equation}
where $d=\min\{m,n\}$ and $\{\ket{i}_{A}\}$ and $\{\ket{i}_{B}\}$ are some
orthonormal bases (for instance the standard ones) in
$\mathcal{H}_{A}$ and $\mathcal{H}_{B}$, respectively. The reason why this state
is called maximally entangled will become clear when we introduce
entanglement measures.

For pure states, the  {\it separability problem} ---
the task of judging if a given quantum state is separable --- is easy
to handle using the concept of
Schmidt decomposition which we introduce in the following theorem.
\begin{thm}\label{SchmidtDec}
Every pure state $\ket{\psi_{AB}}\in \mathcal{H}_{AB}$ with $m\leq
n$ admits the following decomposition
%
\begin{equation}\label{Schmidt0}
|\psi_{AB}\rangle = \sum _{i=1}^{r} \lambda_i |e_i\rangle \otimes
|f_i\rangle
\end{equation}
called also the Schmidt decomposition,
where the local vectors $|e_i\rangle$ and $|f_i\rangle$ form parts of
orthonormal bases in $\mathcal{H}_{A}$ and $\mathcal{H}_{B}$,
respectively. Then, $\lambda_i$ are some positive numbers that satisfy $\sum_{i=1}^{r} \lambda_i^2=1$, and
$r\leq m$.
\end{thm}

%
%
%

The proof of the Theorem \ref{SchmidtDec} employs the singular value
decomposition of the matrix describing the coefficients one gets
by expanding the state in arbitrary orthonormal bases from Alice's
and Bob's Hilbert spaces. The numbers $\lambda_{i}>0$
$(i=1,\ldots,r)$ and $r$ are called, respectively, \textit{the
Schmidt coefficients} and \textit{the Schmidt rank} of
$\ket{\psi_{AB}}$. It is also worth noticing that
$\{\lambda_{i}^{2},\ket{e_{i}}\}$ and
$\{\lambda_{i}^{2},\ket{f_{i}}\}$ are eigensystems of the density
matrices representing the first and second subsystem of
$\ket{\psi_{AB}}$ and $r$ is their rank.

Now, one immediately realizes that Theorem \ref{SchmidtDec} provides a very simple
separability criterion for bipartite pure states: a state $\ket{\psi_{AB}}$
is separable if, and only if its Schmidt rank is one.
Moreover, this criterion is operational, i.e., to check if a given pure state is separable, it suffices
to determine the rank $r$ of one of its subsystems: if $r=1$ (the corresponding subsystem is in a pure
state) then $\ket{\psi_{AB}}$ is separable; otherwise it is
entangled. Note that the maximally entangled state
(\ref{MaximallyEntState}) is already written in the form
(\ref{Schmidt0}), with $r=d$ and all the Schmidt coefficients equal to
$1/\sqrt{d}$.


\subsection{Bipartite mixed states: Separable and entangled states}


Let us now pass to the case of mixed states. Having learned the definition of
separability for pure states, one could naively expect that mixed separable states
are those taking the product form $\rho_A\otimes\rho_B$. This intuition is, however,
not entirely correct and one can argue that all convex combinations of such product states should also
be called separable. This is why the separability problem for mixes states
complicates considerably.


In order to recall the definition of mixed separable states --- first formalized by Werner in 1989
\cite{Werner1989} --- in more precise terms let us consider the following state preparation procedure.
%
%
Imagine that in their distant laboratories, Alice and Bob can produce and manipulate any physical system. Apart from that they can
also communicate using a classical channel (for instance a phone line), however, they
are not allowed to communicate quantumly, meaning that Alice is not allowed to send any
quantum particle to Bob and \textit{vice versa}. These two capabilities, i.e.,
\textit{local operations (LO)} and \textit{classical communication (CC)}, are frequently referred to as LOCC.

Now, let us suppose that in their local laboratories Alice and Bob
can prepare one of $K$ different states $\ket{e_i}\in\mathcal{H}_A$ and
$\ket{f_i}\in\mathcal{H}_B$ $(i=1,\ldots,K)$, respectively. Let us
then assume that in each round of the preparation scheme, Alice
generates with probability $p_{k}$ an integer $k$
$(k=1,\ldots,K)$, which she later sends to Bob using the classical
channel they share. Upon receiving $k$, Alice and Bob use their
local devices to prepare the states $\ket{e_k}$ and $\ket{f_k}$,
respectively. The state that Alice and Bob share after repeating
the above procedure many times is of the form
%
%
%
\begin{eqnarray}
\label{eqn:sepmixed} \varrho_{AB} = \sum_{i=1}^K  p_i
|e_i\rangle\! \langle e_i| \otimes |f_i\rangle\! \langle f_i|,
\end{eqnarray}
which is the aforementioned convex combination of product states.
This is also the most general state that
can be prepared by means of LOCC provided that initially no other
quantum state was shared by Alice and Bob. This gives us the formal
definition of separability \cite{Werner1989}.

\begin{defnn}\label{SepMixed}
A mixed state $\varrho_{AB}$ acting on
$\mathcal{H}_{AB}$ is called \textit{separable} if, and only
if it admits the decomposition (\ref{eqn:sepmixed}).
Otherwise, it is called \textit{entangled}.
\end{defnn}


It then follows from this definition that entangled states cannot be prepared
locally by two parties even if they are allowed to communicate over a classical
channel. To prepare entangled states the physical systems must be
brought together to interact\footnote{Due to entanglement swapping \cite{Zukowski93}, one must
suitably enlarge the notion of preparation of entangled states.
So, an entangled state between two particles can be prepared if
and only if either the two particles (call them $A$ and $B$)
themselves come together to interact at a time in the past, or two
\textit{other} particles (call them $C$ and $D$) do the same, with $C$
having interacted beforehand with A and $D$ with $B$.}. Mathematically, a
non-product unitary operator (i.e., not of the form $U_A\otimes U_B$)
%
%
must \textit{necessarily} act on the physical system
to produce an entangled state from an initial separable one.

Let us recall that the number of pure separable states $K$ necessary to decompose any
separable state into a convex combination of pure product states according to Eq.\ (\ref{eqn:sepmixed})
is limited by the Carath\'eodory theorem as \(K \leq (nm )^2\) (see
\cite{Horodeccy,Horodecki97}). No better bound is known in general, however,
for two-qubit ($\mathcal{H}_{AB}=\mathbbm{C}^2\ot\mathbbm{C}^2$) and qubit-qutrit ($\mathcal{H}_{AB}=\mathbbm{C}^2\ot \mathbbm{C}^3$) systems
it was shown that $K\leq 4$ \cite{a} and
$K\leq 6$ \cite{b}, respectively.

The question whether a given bipartite state is separable or not
turns out to be very complicated (see, e.g., Refs. \cite{Horodeccy,GuhneTothReview}).
Although the general answer to
the separability problem still eludes us, there has been
significant progress in recent years, and we will review some such
directions in the following paragraphs.

\subsection{Entanglement criteria}
\index{Entanglement criteria}


An operational necessary and sufficient criterion for detecting
entanglement still does not exist (see, nevertheless, Ref. \cite{Horodecki96} for a non-operational one). However, over the years the
whole variety of sufficient criteria allowing for detection of entanglement
has been worked out. Below we review one of them, while for others the reader is referred to Ref. \cite{GuhneTothReview}. Note that, even if such an operation necessary and sufficient condition is missing, there are
numerical checks of separability: one can test separability of a state using, for instance, semi-definite programming~\cite{Doherty02,Hulpke05}.
In general --- without a restriction on dimensions --- the
separability problem belongs to the NP-hard class of computational
complexity \cite{Gurvits03}.

{\it Partial transposition} is an easy--to--apply necessary criterion based on
the transposition map first recognized by Choi \cite{Choi82} and then
independently formulated in the separability context by
Peres \cite{Peres96}.

\begin{defnn}\index{Partial transposition}
Let \(\varrho_{AB}\) be a state acting on
\({\cal H}_{AB}\) and let
$T:\mathcal{B}(\mathbbm{C}^{d})\to\mathcal{B}(\mathbbm{C}^{d})$ be the
transposition map with respect to some real basis $\{\ket{i}\}$ in
$\mathbbm{C}^{d}$ defined through $T(X)\equiv
X^{T}=\sum_{i,j}x_{ij}\ket{j}\!\bra{i}$ for any
$X=\sum_{i,j}x_{ij}\ket{i}\!\bra{j}$ from
$\mathcal{B}(\mathbbm{C}^{d})$. Let us now consider an extended map
$T\ot I_{B} $ called hereafter {\it partial transposition} with
$I_{B}$ being the identity map acting on the second subsystem. When
applied to $\varrho_{AB}$, the map $T\ot I_{B} $ transposes the
first subsystem leaving the second one untouched. More formally,
writing $\varrho_{AB}$ as
\begin{equation}
\varrho_{AB} = \sum_{i,j=1}^{m} \sum_{\mu,\nu=1}^{n}
\varrho_{ij}^{\mu\nu}|i\rangle\!\langle j|\otimes
|\mu\rangle\!\langle \nu|,
\end{equation}
where \(\{|i\rangle\}\) and \(\{|\mu\rangle\}\) are real bases
in Alice and Bob Hilbert spaces, respectively, we have
\begin{equation}
\label{eq_partial_trans} (T\ot I_{B}
)(\varrho_{AB})\equiv\varrho_{AB}^{T_A} = \sum_{i,j=1}^{m}
\sum_{\mu,\nu=1}^{n} \varrho_{ij}^{\mu\nu} |j\rangle\!\langle i|
\otimes |\mu\rangle\!\langle \nu|.
\end{equation}
\end{defnn}

%
%
%
%

In an analogous way one defines partial transposition with respect to
Bob's subsystem, denoted by $\varrho_{AB}^{T_{B}}$. Although
the partial transposition of \(\varrho_{AB}\) depends upon the
choice of the basis in which \(\varrho_{AB}\) is written, its
eigenvalues are basis independent. The applicability of the
transposition map in the separability problem can be formalized by
the following statement \cite{Peres96}.

\begin{thm}\label{PeresTh}
For every separable state \(\rho_{AB}\) acting on $\mathcal{H}_{AB}$,
\(\rho_{AB}^{T_{A}}\ \ge\ 0\) and \(\rho_{AB}^{T_{B}}\ \ge\ 0\).
\end{thm}
\begin{proof}
It follows from Definition \ref{SepMixed} that
by applying the partial transposition with respect to the first subsystem to a separable
state $\rho_{AB}$, one obtains
%
%
\begin{eqnarray}\label{Proof}
\rho_{AB}^{T_{A}} = \sum_{i=1}^{K}\ p_i\left(\proj{e_i}
\right)^{T_{A}}\otimes \proj{f_i}
= \sum_{i=1}^{K}\ p_i \proj{e_i^*} \otimes \proj{f_i}, 
\end{eqnarray}
where the second equality follows from the fact that $A^{\dagger}=\left(A^*\right)^{T}$
for all $A$. From the above one infers that $\rho_{AB}^{T_{A}}$
is a proper (and in particular separable) state, meaning that
$\rho_{AB}^{T_{A}}\geq 0$. The same reasoning shows that
$\rho_{AB}^{T_{B}}\geq 0$, which completes the proof.
\end{proof}


%
%

%
Due to the identity
$\varrho_{AB}^{T_{B}}=(\varrho_{AB}^{T_{A}})^{T}$, and the fact
that global transposition does not change eigenvalues, partial
transpositions with respect to the $A$ and $B$ subsystems are
equivalent from the point of view of the separability problem.

In conclusion, we have a simple criterion, called \textit{partial
transposition criterion}, for detecting entanglement:
if the spectrum of one of the partial transpositions of
$\varrho_{AB}$ contains at least one negative eigenvalue then
$\varrho_{AB}$ is entangled. As an example, let us apply the
criterion to pure entangled states. If $\ket{\psi_{AB}}$ is
entangled, it can be written as (\ref{Schmidt0}) with $r>1$. Then,
the eigenvalues of $\proj{\psi_{AB}}^{T_{A}}$ are
$\lambda_{i}^{2}$ $(i=1,\ldots,r)$ and $\pm\lambda_{i}\lambda_{j}$
$(i\neq j\; i,j=1,\ldots,r)$. So, an entangled $\ket{\psi_{AB}}$
of Schmidt rank $r>1$ has partial transposition with $r(r-1)/2$
negative eigenvalues violating the criterion stated in Theorem
\ref{PeresTh}. Note that in systems of two qubits or a qubit and a
qutrit the partial transposition criterion provides the
necessary and sufficient condition for separability
\cite{Horodecki96}. This is no more true in higher dimensions, due
to the existence of entangled states with positive partial transposition
\cite{Horodecki97,Horodecki98}.

\subsection{Entanglement measures}

Although the separability criterion discussed above allows one to check
whether a given state $\rho_{AB}$ is entangled,
it does not tell us (at least not directly) how much entanglement
it has. Such a quantification is necessary
because entanglement is a resource in quantum
information theory.
%
%
There are several complementary ways to
quantify entanglement of bipartite quantum states (see
\cite{Bennett96b,Vedral97,DiVincenzo98,Laustsen03,Nielsen99,Vidal00,Jonathan99,Horodecki04,Horodecki01b,Plenio07,Horodeccy}
and references therein) and in what follows we briefly
discuss one of them.

Let us now introduce the definition of entanglement
measures (for a more detailed axiomatic description, and
other properties of entanglement measures, the reader is
encouraged to consult, e.g., \cite{Horodeccy,Horodecki01b,Plenio07}). The main ingredient in this definition
is the monotonicity under LOCC operations. More
precisely, if $\Lambda$ denotes some LOCC operation, and $E$ is our
candidate for the entanglement measure, $E$ has to satisfy
\begin{equation}\label{monot1}
E(\Lambda(\varrho))\leq E(\varrho),
\end{equation}
i.e., it should not increase under LOCC operations. Another requirement says that $E$ vanishes on separable states. At this point it is worth noticing that from the monotonicity under LOCC operations (\ref{monot1}) it already follows that $E$ is constant and minimal on separable states and also that it is invariant under unitary operations (see Ref.\ \cite{Horodeccy}).
%
%
%

\subsection{Von Neumann entropy}

A ``good'' entanglement measure for a pure state
$|\psi_{AB}\rangle$ is the von Neumann entropy of the density matrix describing one of its subsystems, say
the first one which arises by tracing out Bob's subsystem of
$\ket{\psi_{AB}}$, i.e., $\varrho_A = {\rm Tr}_B\proj{\psi_{AB}}$. Recalling then that
the von Neumann entropy of a density matrix
$\rho$ is defined through $S(\rho) = -{\rm Tr}(\rho \log \rho)$, the following
quantity
%
%
\begin{equation}
     E(\ket{\psi_{AB}}) = S(\varrho_A) = S(\varrho_B)=-\sum_i\lambda_i^2\log\lambda_i^2,
\end{equation}
was shown to be an entanglement measure
\cite{Vidal00}. Notice that for the maximally entangled states
(\ref{MaximallyEntState}) one has $E(\ket{\psi_{+}^{(d)}}) =\log
d$. On the other hand, $E$ is an entanglement measure only for pure states. Separable mixed states have classical correlations, and thus the non-zero entropy of the reduced density matrix. In the following we will concentrate on the entanglement properties of the ground states of many-body
systems. There the von Neumann entropy of a density matrix
reduced to some region $R$ will play a fundamental role.

\section{Entanglement in Many-Body Systems}
\label{ManyBody}

\subsection{Computational complexity}

Let us start this discussion by considering simulations of quantum
systems with classical computers. What can be simulated
classically  \cite{Lewenstein2012}? The systems that can be
simulated classically are those to which we can apply efficient
numerical methods, such as the quantum Monte Carlo method that works,
for instance, very well for bosonic unfrustrated systems.
Sometimes we may apply systematic perturbation theory, or even use
exact diagonalization for small systems (say, for frustrated
antiferromagnets consisting of 30-40 spins 1/2). There is a
plethora of variational and related methods available, such as
various mean field methods, density functional theory (DFT),
dynamical mean field theory (DMFT), and methods employing tensor
network states  (TNS), such as Matrix-Product States (MPS),
Projected-Entangled-Pair States (PEPS), Multi-scale Entanglement
Renormalization Ansatz (MERA), etc.

What is then computationally hard? Generic examples include
fermionic models, frustrated systems, or disordered systems. While
MPS techniques allow for efficient calculation of the ground states
and also excited states in 1D, there are, even in 1D, no efficient algorithms
to describe the out-of-equilibrium quantum dynamics. Why do we
still have hopes to improve our classical simulation  skills in
the next future? This is connected with the recent developments of
the tensor network states and observation that most of the states
of physical interest, such as the ground states of local
Hamiltonians, are non generic and fulfill the, so called, area
laws.

\subsection{ Entanglement of a generic state}

Before we turn to the area laws for physically relevant states let
us first consider a {\it generic} pure state in  the Hilbert space
in $\mathbbm{C}^{m}\ot \mathbbm{C}^{n}$ ($m\leq n$). Such a generic
state (normalized) has the form
\begin{equation}
|\Psi\rangle=\sum_{i=1}^m\sum_{j=1}^n
\alpha_{ij}|i\rangle|j\rangle,
\end{equation}
where $\{\ket{i}\ket{j}\}$ is the standard basis in $\mathbbm{C}^{m}\ot \mathbbm{C}^{n}$ and the complex numbers $\alpha_{ij}$ may be regarded as random variables distributed uniformly on a hypersphere, i.e., distributed
according to the probability density 
\begin{equation}
P(\alpha)\propto \delta\left(\sum_{i=1}^m\sum_{j=1}^n
|\alpha_{ij}|^2-1\right), \label{distalpha}
\end{equation}
with the only constraint being the normalization. As we shall see,
such a generic state fulfills on average a ``volume" rather than
an area law. To this aim we introduce a somewhat more rigorous
description, and we prove that on average, the entropy of one of
subsystems of bipartite pure states in $\mathbbm{C}^{m}\ot
\mathbbm{C}^{n}$ ($m\leq n$) is almost maximal for sufficiently
large $n$. In other words, typical pure states in
$\mathbbm{C}^{m}\ot \mathbbm{C}^{n}$ are almost maximally entangled.
This ``typical behavior" of pure states happens to be completely
atypical for ground states of local Hamiltonians with an energy
gap between ground and first excited eigenstates. More precisely, one has the following theorem
(see, e.g., Refs. \cite{Lubkin,LloydPagels,Page,BengtssonBook,Foong,SenS,Sanchez-Ruiz}).

\begin{thm}\label{typical}
Let $\ket{\psi_{AB}}$ be a bipartite pure state from
$\mathbbm{C}^{m}\ot\mathbbm{C}^{n}$ $(m\leq n)$ drawn at random
according to the Haar measure on the unitary group and
$\varrho_{A}=\tr_{B}\proj{\psi_{AB}}$ be its subsystem acting on
$\mathbbm{C}^{m}$. Then,
\begin{equation}\label{EntrAppr0}
\langle S(\varrho_{A})\rangle \simeq \log m-\frac{m}{2n}.
\end{equation}
\end{thm}

%
%

%

Notice that the above  result   can
be estimated very easily  by relaxing the normalization constraint in the
distribution (\ref{distalpha}), and replacing it by a product of
independent Gaussian distributions,
$P(\alpha)=\prod_{i,j}(nm/\pi)\exp[-nm|\alpha_{ij}|^2]$, with
$\langle \alpha_{ij}\rangle=0$, and $\langle
|\alpha_{ij}|^2\rangle=1/nm$. According to the central limit theorem, the latter distribution
tends for $nm\to \infty$ to a Gaussian one for $\sum_{i=1}^m\sum_{j=1}^n
|\alpha_{ij}|^2$ centered at 1 of width $\simeq 1/\sqrt{nm}$.
One then straightforwardly  obtains that $\langle {\rm
tr}\varrho_A\rangle=1$, and after a little more tedious
calculation that $\langle {\rm tr}\varrho^2_A\rangle=(n+m)/nm$, which
agrees asymptotically with the above result for $nm\gg 1$.

\section{Area laws}
\label{Area}

Generally speaking, area laws mean that, when we consider a large region $R$ of a large system $L$ in a pure state, some of the physical properties of $R$ such as the von Neumann entropy of the reduced density matrix $\rho_R$ representing it will depend only on the boundary $\partial R$ (cf. Fig. \ref{FigureAreaLaw}).

\begin{figure}[t!] 
   \centering
   \includegraphics[width=0.70\textwidth]{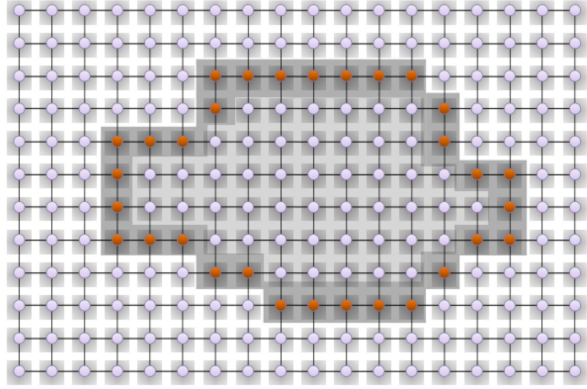}
   \caption{Schematic representation of a lattice system $L$, an arbitrary region
   $R$ (denoted in light grey background), and its boundary $\partial R$ (denoted
   in dark grey background).}
   \label{FigureAreaLaw}
\end{figure}

\subsection{Quantum area laws in 1D}

Let us start with the simplest case of one-dimensional
lattices,  $L=\{1,\ldots,N\}$. Let $R$ be a subset of $L$
consisting of $n$ contiguous spins starting from the first site,
i.e., $R=\{1,\ldots,n\}$ with $n<N$. In this case the boundary
$\partial R$ of the region $R$ contains one spin for open boundary conditions, and
two for periodic ones. Therefore, in this case the area law is extremely simple:
\begin{equation}\label{1DAreaLaw}
S(\varrho_{R})=O(1).
\end{equation}

The case of $D=1$ seems to be quite well understood. In general,
all local gapped systems (away from criticality) satisfy the above
law, and there might be a logarithmic divergence of entanglement
entropy when the system is critical. To be more precise, let us
recall the theorem by Hastings leading to the first of the above
statements, followed by examples of critical systems showing a
logarithmic divergence of the entropy with the size of $R$.

Consider the nearest-neighbor interaction Hamiltonian
\begin{equation}\label{OneDHamilt}
H=\sum_{i\in L}H_{i,i+1},
\end{equation}
where each $H_{i,i+1}$ has a nontrivial support only on the sites
$i$ and $i+1$. We assume also that the operator norm of all the
terms in Eq. (\ref{OneDHamilt}) are upper bounded by some positive
constant $J$, i.e., $\|H_{i,i+1}\|\leq J$ for all $i$ (i.e., we
assume that the interaction strength between $i$th site and its
nearest-neighbor is not greater that some constant). Under these
assumptions, Hastings proved the following theorem \cite{HastingsJSM}.

\begin{thm}\label{AreaLaw1}
Let $L$ be a one-dimensional lattice with $N$ $d$-dimensional sites,
and let $H$ be a local Hamiltonian (\ref{OneDHamilt}).
Assuming that $H$ has a unique ground state separated from the
first excited states by the energy gap $\Delta E>0$, the entropy
of any region $R$ satisfies
\begin{equation}
S(\varrho_{R})\leq 6c_{0}\xi2^{6\xi\log d}\log\xi\log d
\end{equation}
with $c_{0}$ denoting some constant of the order of unity and
$\xi=\min\{2v/\Delta E,\xi_{C}\}$. Here, $v$ denotes the sound
velocity and is of the order of $J$, while $\xi_{C}$ is a length scale of
order unity.
\end{thm}

%
%
%

Let us remark that both constants appearing in the above theorem
come about from the Lieb-Robinson bound \cite{LiebRobinson} (see also
Ref.\ \cite{Masanes} for a recent simple proof of this bound).
This theorem tells us that when the one-dimensional system with
the local interaction defined by Eq.\ (\ref{OneDHamilt}) is away
from the criticality $(\Delta E>0)$, the entropy of $R$ is bounded
by some constant independent of $|R|$. One can naturally ask if there
exist gapped systems with long-range interaction violating
(\ref{1DAreaLaw}). This was answered in the affirmative in Ref.
\cite{Dur05,EisertOsborne}, which gave examples of one-dimensional
models with long--range interactions, nonzero energy gap, and
scaling of entropy diverging logaritmically with $n$.

The second question one can ask is about the behavior of the
entropy when $\Delta E\to 0$ and the system
becomes critical. Numerous analytical and numerical results show
that usually one observes a logarithmic divergence of
$S(\varrho_{R})$ with the size of the region $R$ (we refer the reader to recent reviews
\cite{EisertRMP,Latorre09}, and to the special issue of
J. Phys. A devoted to this subject \cite{CardyJPA}).

Concluding, let us mention that there is an extensive literature
on the logarithmic scaling of the block entropy using conformal
field theory methods (see Ref.\ \cite{Calabrese09} for a very good
overview of these results). Quite
generally, the block entropy at criticality scales as
\begin{equation}
S(\varrho_{R})=\frac{c}{3}\log_{2}\left(\frac{|R|}{a}\right)+O(1),
\end{equation}
or, more in general for the R\'enyi entropy\footnote{Recall that
the quantum R\'enyi entropy is defined as
%
$S_{\alpha}=\log_{2}\left[\Tr\left(\varrho^{\alpha}\right)\right]/(1-\alpha)$
%
where $\alpha\in[0,\infty]$. For $\alpha=0$ one has
$S_{0}(\varrho)=\log_{2}\mathrm{rank}(\varrho)$ and
$S_{\infty}=-\log_{2}\lambda_{\mathrm{max}}$ with
$\lambda_{\mathrm{max}}$ being the maximal eigenvalue of
$\varrho$. }
\be
S_{\alpha}(\varrho_{R})=\frac{c}{6}\left(1+\frac{1}{\alpha}\right)\log_{2}\left(\frac{|R|}{a}\right)+O(1),
\ee
where $c$ is called the {\it central charge} of the
underlying conformal field theory, and $a$ is the cutoff parameter
(the lattice constant for lattice systems).
Recently, these results were generalized in Ref. \cite{Brandao}, where the authors derived the area laws  only from the assumption of the exponential decay of correlations, and without any assumption about the gap.

\subsection{Higher--dimensional systems}

The situation is much more complex in higher spatial dimensions
$(D>1)$. The boundary $\partial R$ of the general area law, Eq.
(\ref{AreaLawGen}), is no longer a simple one or two--element set
and can have a rather complicated structure. Even if there are no
general rules discovered so far, it is rather believed that
\begin{equation}\label{AreaLawGen}
S(\varrho_{R})=O(|\partial R|).
\end{equation}
holds for ground states of local gapped
Hamiltonians. This intuition is supported by results showing that
for quadratic quasifree fermionic and bosonic lattices the area
law (\ref{AreaLawGen}) holds \cite{EisertRMP}.
Furthermore, for critical fermions the entropy of a cubic region
$R=\{1,\ldots,n\}^{D}$ is bounded  as $\gamma_1
n^{D-1}\log_{2}n\leq S(\varrho_{R})\leq \gamma_2
n^{D-1}(\log_{2}n)^{2}$ with $\gamma_{i}$ $(i=1,2)$ denoting some
constants  \cite{WolfPRL06,GioevKlich06,FarkasZimboras07}.
Notice that the proof of this relies on the fact that the
logarithmic negativity\footnote{Negativity and logarithmic negativity
are entanglement measures based on partial transpose. The first one is defined as
$N(\varrho_{AB})=(1/2)(\|\varrho_{AB}^{T_B}\|-1)$ \cite{Zyczkowski98,Vidal02b}. The calculation of $N$ even for mixed states reduces to determination of eigenvalues of $\varrho_{AB}^{T_{B}}$, and
amounts to the sum of the absolute values of negative eigenvalues
of $\varrho_{AB}^{T_{B}}$.
Then, the logarithmic negativity is defined as
$E_{N}(\varrho_{AB})=\log_2\|\varrho^{\Gamma}_{AB}\|=\log_2[2N(\varrho_{AB})+1]$ \cite{Vidal02b}. }
%
upper bounds the von Neumann
entropy, i.e., for any $\ket{\psi_{AB}}$, the
inequality $S(\varrho_{A(B)})\leq E_{N}(\ket{\psi_{AB}})$ holds.
This in turn is a consequence of monotonicity of the R\'enyi
entropy $S_{\alpha}$ with respect to the order $\alpha$, i.e.,
$S_{\alpha}\leq S_{\alpha'}$ for $\alpha\geq \alpha'$. This is one
of the numerous instances, where insights from quantum information help to
deal with problems in many--body physics.

Recently, Masanes \cite{Masanes} showed that in
the ground state (and also low--energy eigenstates) the entropy of a
region $R$ (even a disjoint one) always scales at most as the size
of $|\partial R|$ with some correction proportional to
$(\log|R|)^{D}$ --- as long as the Hamiltonian $H$ is of the local
form
\begin{equation}
H=\sum_{i\in L}H_{i},
\end{equation}
where each $H_{i}$ has nontrivial support only on the
nearest-neighbors of the $i$th site, and, as before, satisfies
$\|H_{i}\|\leq J$ for some $J>0$. Thus, the behavior of entropy
which is considered to be a violation of the area law, can in fact be treated as an area law itself.
This is because in this case\footnote{It should be noticed that
one can have much stronger condition for such scaling of entropy.
To see this explicitly, say that $R$ is a cubic region
$R=\{1,\ldots,n\}^{D}$ meaning that $|\partial R|=n^{D-1}$ and
$|R|=n^{D}$. Then since $\lim_{n\to\infty}[(\log
n)/n^{\epsilon}]=0$ for any (even arbitrarily small) $\epsilon>0$,
one easily checks that $S(\varrho_{R})/|\partial
R|^{1+\epsilon}\to 0$ for $|\partial R|\to \infty$.}
$[|\partial R|(\log|R|)^{k}]/|R|\to 0 $ for $|R|\to \infty$ with
some $k>0$, meaning that this behaviour of entropy is still very
different from the typical one that follows from Theorem
\ref{typical}. That is, putting $m=d^{|R|}$ and $n=d^{|L\setminus
R|}$ with $|L|\gg|R|$, one has that $S(\varrho_{R})/|R|$ is
arbitrarily close to $\log d$ for large $|R|$. More precisely,
the following theorem was proven in Ref. \cite{Masanes}.

\begin{thm}
Let $R$ be some arbitrary (even disjoint) region of $L$. Then,
provided that certain "natural" bounds on correlation functions (polynomial decay with distance)
and on the density of states (number of eigenstates of the Hamiltonian limited to $R$
with energies smaller than $e$ is exponentially
bounded by $|R|^{\gamma(e-e_0)}$, where $\gamma$ is a constant, and $e_0$ is
the lowest energy) hold, the entropy of the reduced density matrix $\varrho_{R}$ of
the ground state of $H$ satisfies
\begin{equation}\label{formula}
S(\varrho_{R})\leq C|\partial R|(10\xi\log|R|)^{D}
+O(|\partial R|(\log|R|)^{D-1}),
\end{equation}
where $C$ collects the constants $D,\xi,\gamma,J,\eta$, and $d$.
If $R$ is a cubic region, the above statement simplifies, giving
$S(\varrho_{R})\leq \widetilde{C}|\partial R|\log|R|+O(|\partial
R|)$ with $\widetilde{C}$ being some constant.
\end{thm}

%
%

\subsubsection{Area laws for mutual information -- classical and quantum Gibbs states}

So far, we considered area laws only for ground states of local
Hamiltonians. In addition, it would be very interesting to ask
similar questions for nonzero temperatures. Here, however, one
cannot rely on the entropy of a subsystem, as in the case of mixed
states it is no longer an entanglement measure. Instead, one can use the
{\it quantum mutual information}
which measures the total amount of correlation in bipartite quantum systems
\cite{GroismanPopescuWinterPRA05}. It is defined as
\begin{equation}\label{Mutual}
I(A:B)=S(\varrho_{A})+S(\varrho_{B})-S(\varrho_{AB}),
\end{equation}
where $\varrho_{AB}$ is some bipartite state and
$\varrho_{A(B)}$ stand for its subsystems.
It should be noticed that for pure states the
mutual information reduces to twice the amount of entanglement of
the state.

Recently, it was proven that thermal states
$\varrho_{\beta}=e^{-\beta H}/\tr[e^{-\beta H}]$ with local
Hamiltonians $H$ obey an area law for mutual information.
Interestingly, a similar conclusion was drawn for classical
lattices, in which we have a classical spin with the
configuration space $\mathbbm{Z}_{d}$ at each site, and instead of density
matrices one deals with probability distributions. In the following
we review these two results, starting from the classical
case.

To quantify correlations in classical systems, we use the
classical mutual information,  defined as in Eq.\ (\ref{Mutual})
with the von Neumann entropy substituted by the Shannon entropy
$H(X)=-\sum_{x}p(x)\log_2 p(x)$, where $p$ stands for a
probability distribution characterizing random variable $X$. More
precisely, let $A$ and $B=S\setminus A$ denote two subsystems of
some classical physical system $S$. Then, let $p(x_A)$ and
$p(x_B)$ be the marginals of the joint probability distribution
$p(x_{AB})$ describing $S$ ($x_{a}$ denotes the possible
configurations of subsystems $a=A,B,AB$). The correlations between
$A$ and $B$ are given by the \textit{classical mutual information}
\begin{equation}\label{MutualInformation}
I(A:B)=H(A)+H(B)-H(AB).
\end{equation}
We are now ready to recall the results of
\cite{Wolf08}.

\begin{thm}
Let $L$ be a lattice with $d$--dimensional classical spins at each
site. Let $p$ be a Gibbs probability distribution coming from
finite--range interactions on $L$. Then, dividing $L$ into regions
$A$ and $B$, one has
\begin{equation}
\label{AreaLawClassicalMutInf}
I(A:B)\leq |\partial A|\log d.
\end{equation}
\end{thm}

%
%

Let us now show that a similar conclusion can be drawn in the case
of quantum thermal states \cite{Wolf08}, where the Markov property
does not hold in general.

\begin{thm}
Let $L$ be a lattice consisting of $d$-dimensional quantum systems
divided into parts $A$ and $B$ ($L=A\cup B$). Thermal states
$(T>0)$ of local Hamiltonians $H$ obey the following area law
\begin{equation}
\label{MutualAreaLaw}
I(A:B)\leq \beta
\tr[H_{\partial}(\varrho_{A}\ot\varrho_{B}-\varrho_{AB})].
\end{equation}
where $H_{\partial}$ stands for
interaction terms connecting these two regions.
\end{thm}

%

Let us notice that the right--hand side of Eq.
(\ref{MutualAreaLaw}) depends only on the boundary, and therefore
it gives a scaling of mutual information similar to the classical
case (\ref{AreaLawClassicalMutInf}). Moreover, for the nearest-neighbor
interaction, Eq.\ (\ref{MutualAreaLaw}) simplifies to $I(A:B)\leq
2\beta\|h\|\,|\partial A|$ with $\|h\|$ denoting the largest
eigenvalue of all terms of $H$ crossing the boundary.

\subsection{The world according to tensor networks}

Quantum many-body systems are, in general, difficult to describe:
specifying an arbitrary state of a system with $N$ two-level
subsystems requires $2^N$ complex numbers. For a classical
computer, this presents not only storage problems, but also
computational ones, since simple operations like calculating the
expectation value of an observable would require an exponential
number of operations. However, we know that completely separable
states can be described with about $N$ parameters --- indeed, they
correspond to classical states. Therefore, what makes a quantum
state difficult to describe are quantum correlations, or
entanglement. We saw already that even if in general the entropy
of a subsystem of an arbitrary state is proportional to the
volume, there are some special states which obey an entropic area
law. Intuitively, and given the close relation between entropy and
information, we could expect that states that follow an area law
can be described (at least approximately) with much less
information than a general state. We also know that such low
entanglement states are few, albeit interesting --- we only need
an efficient and practical way to describe and parametrize them
\footnote{Note, however, that an area law does not imply an
efficient classical parametrization (see, e.g.,
Ref.~\cite{eisert_new}.}.

Consider a general pure state of a system with $N$ $d$-level particles,
\begin{equation}
| \psi \rangle = \sum_{i_1,i_2,\ldots,i_N=1}^{d} c_{i_1i_2\ldots
i_N} | i_1,i_2,\ldots,i_N \rangle.
\end{equation}
When the state has no entanglement, then $c_{i_1i_2\ldots
i_N}=c^{(1)}_{i_1}c^{(2)}_{i_2}\ldots c^{(N)}_{i_N}$ where all
$c$'s are scalars. The locality of the information (the set of
coefficients $c$ for each site is independent of the others) is
key to the efficiency with which separable states can be
represented. How can we keep this locality while adding complexity
to the state, possibly in the form of correlations but only to
nearest-neighbors? As we shall see, we can do this by using a
tensor at each site of our lattice, with one index of the tensor
for every physical neighbor of the site, and another index for the
physical states of the particle. For example, in a one-dimensional
chain we would assign a matrix for each state of each particle,
and the full quantum state would be written as
\begin{equation}
| \psi \rangle = \sum_{i_1,i_2,\ldots,i_N=1}^{d} {\mathrm{Tr}} \left[
A^{[1]}_{i_1}A^{[2]}_{i_2}\ldots A^{[N]}_{i_N} \right] |
i_1,i_2,\ldots i_N \rangle, \label{mps}
\end{equation}
where $A^{[k]}_{i_k}$ stands for a matrix of dimensions $D_k
\times D_{k+1}$. A useful way of understanding the motivation for
this representation is to think of a valence bond picture
\cite{Verstraete04a}. Imagine that we replace every particle at
the lattice by a pair (or more in higher dimensions) of particles
of dimensions $D$ that are in a maximally entangled state with
their corresponding partners in a neighboring site (see Figure
\ref{FigureMPS}). Then, by applying a map from these virtual
particles into the real ones,
\begin{equation}
{\cal A}=\sum_{i=1}^d \sum_{\alpha,\beta=1}^D
A_{\alpha,\beta}^{[i]} |i\rangle\!\langle\alpha ,\beta |,
\end{equation}
we obtain a state that is expressed as Eq.\ (\ref{mps}). One can
show that any state $| \psi \rangle \in (\mathbbm{C}^{d})^{\otimes N} $ can be
written in this way with $D=\max_m D_m \le d^{N/2}$. Furthermore,
a matrix product state can always be found such that
\cite{Vidal03a}
\begin{itemize}
\item $\sum_i A^{\dagger [k]}_{i} A^{[k]}_{i} = \mathbbm{1}_{D_k}$, for $1 \le k \le N$,
\item $\sum_i A^{\dagger[k]}_{i} \Lambda^{[k-1]} A^{[k]}_{i}  = \Lambda^{[k]}$, for $1 \le k \le N$, and
\item For open boundary conditions $\Lambda^{[0]}=\Lambda^{[N]}=\mathbbm{1}$, and $\Lambda^{[k]}$ is a $D_{k+1} \times
D_{k+1}$ positive diagonal matrix, full rank, with ${\rm Tr}
\Lambda^{[k]}=1$.
\end{itemize}
In fact, $\Lambda^{[k]}$ is a matrix whose diagonal components
$\lambda^{k}_n$ $(n=1,\ldots,D_k)$ are the non-zero eigenvalues of
the reduced density matrix obtained by tracing out the particles
from $k+1$ to $N$, i.e., the Schmidt coefficients of a bipartition
of the system at site $k$. An MPS with these properties is said to
be in its canonical form \cite{Perez-Garcia07}.

\begin{figure}[t!] 
   \centering
   \includegraphics[width=0.60\textwidth]{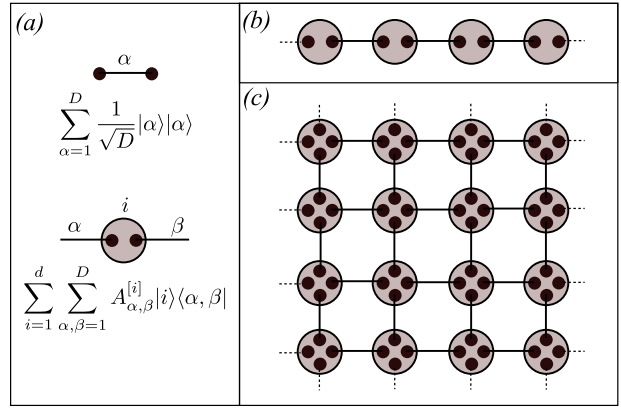}
   \caption{Schematic representation of tensor networks. In panel $(a)$ we show the meaning of the
   elements in the representation, namely the solid line joining two virtual particles in different sites
   means the maximally entangled state between them, and the grey circle represents the map
   from virtual particles in the same site to the physical index.
   In panel $(b)$  we see a one-dimensional tensor network or MPS, while in $(c)$ we show how
   the scheme can be extended intuitively to higher dimensions --- in the two-dimensional example
   shown here, a PEPS that contains four virtual particles per physical site.}
   \label{FigureMPS}
\end{figure}

Therefore, Eq.\ (\ref{mps}) is a representation of all possible
states --- still cumbersome. It becomes an efficient
representation when the virtual bond dimension $D$ is small, in
which case it is typically said that the state has an MPS
representation. In higher dimensions we talk about PEPS
\cite{Verstraete04d}. When entanglement is small (but finite),
most of the Schmidt coefficients are either zero or decay rapidly
to zero \cite{Vidal03a}. Then, if $| \psi \rangle $ contains
little entanglement, we can obtain a very good approximation to it
by truncating the matrices $A$ to a rank $D$ much smaller than the
maximum allowed by the above theorem, $d^{N/2}$. In fact, one can
demonstrate the following fact \cite{Perez-Garcia07}.

\begin{lemma}
 For any pure state $\ket{\psi}$, there exists an MPS $|\psi_D \rangle$ with the bond dimension $D$ such that
 \begin{equation}
\| | \psi \rangle - | \psi_D \rangle \|^{2} < 2
\sum_{k=1}^{N-1} \sum_{i=D+1}^{d^{\min(k, N-k)}}
\lambda^{[k]}_i.
\end{equation}
%
%
 \label{lemmaBound}
\end{lemma}


This Lemma  is most powerful in the context of
numerical simulations of quantum states: it gives a controllable
handle on the precision of the approximation by MPS. In practical
terms, for the representation to be efficient the Schmidt
coefficients $\lambda$ need to decay faster than polynomially.
However, we can be more precise and give bounds on the error of
the approximation in terms of entropies \cite{Schuch07}:

\begin{lemma}
Let $S_\alpha(\rho)=\log [\tr (\rho^\alpha)]/(1-\alpha)$ be the
R\'enyi entropy of a reduced density matrix $\rho$, with $0<
\alpha < 1$. Denote $\epsilon(D) = \sum_{i=D+1}^{\infty}
\lambda_i$, with $\lambda_i$ being the eigenvalues of $\rho$ in
nonincreasing order. Then,
\begin{equation}
\log [\epsilon(D)] \le \textstyle{\frac{1-\alpha}{\alpha}} \left[
S_\alpha(\rho) - \log \left(\textstyle{\frac{D}{1-\alpha}}\right)
\right].
\end{equation}
\end{lemma}
%
%
The question now is when can we find systems with relevant states
that can be written efficiently as a MPS; i.e. how broad is the
simulability of quantum states by MPS. For example, one case of
interest where we could expect the method to fail is near quantum
critical points where correlations (and entanglement) are singular
and might diverge. However, at least in 1D systems, the following fact remains true
\cite{Perez-Garcia07}.

\begin{lemma}
In one dimension there exists a scalable, efficient MPS
representation of ground states even at criticality.
\end{lemma}



\section{Nonlocality in many body systems}
\label{Nonlocality}

Let us now turn to nonlocality in many-body systems. We start by
explaining what the concept of nonlocality means, using the
contemporary language of device independent quantum information
processing (DIQIP). Recent successful hacking attacks on quantum
cryptographic devices stimulated a novel approach to quantum
information theory in which protocols are defined independently of
the inner working of the devices used in the implementation, hence
the term DIQIP.

\subsection{Probabilities and correlations -- DIQIP approach}

\begin{figure}[h!]
\begin{center}
\includegraphics[width=0.35\textwidth]{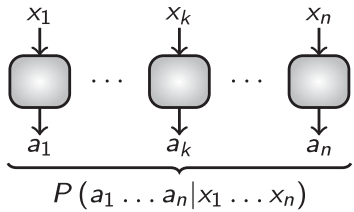}
\end{center}
\caption{Schematic picture of the DIQIP approach.}\label{fig:Bell1}
\end{figure}

The idea of DIQIP is at best explained with the graphical scheme presented on Fig.
\ref{fig:Bell1}. We consider here the following scenario, usually
referred to as the \textit{$(n,m,d)$ scenario}. Let us consider
$n$ spatially separated parties $A_1,\ldots,A_n$ and imagine that
each of them possesses a black box with $m$ buttons representing
the measurement choices (or observables) and $d$ lights
representing the measurement outcomes. Now, in each round of the
experiment every party is allowed to press one of the buttons
causing one of the lights to shine. The only information
accessible in such an experiment is contained in a set of $(md)^n$
conditional probabilities $P(a_1,\ldots, a_n|x_1,\ldots, x_n)$ of
obtaining outputs $a_1, a_2, \ldots, a_n$, provided observables
$x_1, x_2, \ldots, x_n$ were measured. In what follows we enumerate the measurements and outcomes as
$x_i=1,\ldots,m$ and $a_i=0,\ldots,d-1$, respectively.

%
%

The set of all such probability distributions is convex as by mixing any two of them one obtains another probability distribution; in fact, it is a polytope. From the physical point of view (causality, special relativity) the probabilities must fulfil the {\it non-signalling conditions}, i.e., the choice of measurement by the $k$-th party, cannot be signalled to the others. Mathematically it means that for any $k=1,\ldots,n$, the following condition
\begin{eqnarray}
&\sum_{a_k}P(a_1, a_2, \ldots,  a_k, \ldots, a_n|x_1, x_2, \ldots, x_k,\ldots, x_n)\nonumber\\
& = P(a_1, a_2, \ldots,  a_{k-1}, a_{k+1}\ldots, a_n|x_1, x_2, \ldots, x_{k-1},x_{k+1}\ldots, x_n),
\end{eqnarray}
is fulfilled. In other words, the marginal probability distribution describing correlations seen by
the $n$ parties except the $k$th one is independent of $x_k$. We call correlations satisfying the above constraints \textit{nonsignalling correlations}. It is easy to see that they also form a polytope.
Let us also notice that the above conditions together with normalization
clearly reduce the number of independent probabilities. For instance, in the simplest $(2,2,2)$
scenario there are eight independent probabilities out of sixteen and they can be chosen as
$P(0,0|x_1,x_2)$, $P_A(0|x_1)$, and $P_B(0|x_2)$ with $x_1,x_2=1,2$.

The \textit{local} or \textit{classical correlations} are defined via the concept of a local hidden variable $\lambda$. Imagine that the only resource shared by the parties is some classical information $\lambda$ (called also LHV) distributed among them with probability $q_{\lambda}$.
The correlations that the parties are able to establish in such case are
of the form
\begin{equation}
P(a_1,  \ldots, a_n|x_1,  \ldots,  x_n) = \sum_{\lambda} q_\lambda D(a_1|x_1,\lambda)  \ldots D(a_n|x_n,\lambda),
\end{equation}
where $D(a_k|x_k,\lambda)$ are deterministic probabilities, i.e., for any $\lambda$, $D(a_k|x_k,\lambda)$ equals one for some outcome, and zero for
all others. What is important in this expression is that measurements of different parties are independent, so that the probability is a product of
terms corresponding to different parties.

Classical correlations form a convex set which is also a polytope, denoted $\mathbbm{P}$ (cf. Fig. \ref{fig:Bell1}). Its extremal points (or vertices) are the above form, i.e., $\prod_{i=1}^n D(a_i|x_i,\lambda)$
with fixed $\lambda$.
%
%
The famous theorem of John Bell states that the quantum-mechanical
probabilities, which also form a convex set $\mathcal{Q}$, may stick out of the
classical polytope \cite{Bell64}. The quantum probabilities are given by the
trace formula for the set of local measurements
\begin{equation}
P(a_1, \ldots, a_n | x_1, \ldots, x_n)={\rm tr}(\rho M_{a_1}^{x_1}\otimes \cdots \otimes M_{a_n}^{x_n}),
\end{equation}
where $\rho$ is some $n$-partite state and $M_{a_i}^{x_i}$
denote the measurement operators, meaning that $M_{a_i}^{x_i}\geq 0$ for any $a_i$, $x_i$ and
$i$, and
%
%
\begin{equation}
\sum_{a_i}M_{a_i}^{x_i}=\mathbbm{1},
%
\end{equation}
for any choice of the measurement $x_i$ and party $i$. As we do not impose any constraint on the local dimension, we can always choose the measurements to be projective, i.e., the measurement operators additionally satisfy
$M_{a'_i}^{x_i}M_{a_i}^{x_i}=\delta_{a'_i,a_i} M_{a_i}^{x_i}$.

%

\begin{figure}[t!]
\begin{center}
\includegraphics[width=0.7\textwidth]{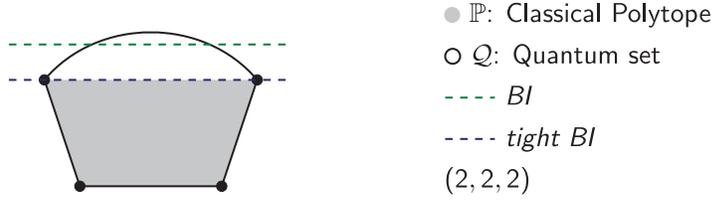}
\end{center}
\caption{Schematic representation of the sets of correlations: classical (grey area) and quantum (the area bounded by the thick line). Clearly, the latter is a superset of the former and it was shown by Bell \cite{Bell64} that they are not equal, i.e., there are quantum correlations that do not fall into the grey area. Then, the black dots represent the vertices of the classical polytope $\mathbbm{P}$ -- deterministic classical correlations. Finally, the dashed lines represent Bell inequalities, in particular, the one represented by the black line is tight, that is, it corresponds to the facet of the classical set.}\label{fig:zbiory}
\end{figure}

The concept of the Bell inequalities is explained in Fig.
\ref{fig:zbiory}. Any hyperplane in the space of probabilities that
separates the classical polytope from the rest is a Bell inequality:
everything, which is above the upper horizontal dashed line is
obviously nonlocal. But, the most useful are the {\it tight Bell
inequalities}, which correspond to the facets of the classical
polytope, i.e. its walls of maximal dimensions (lower horizontal
dashed line). To be more illustrative, let us now present a
particular example of a Bell inequality.
To this end, let us consider the simplest $(2,2,2)$
scenario consisting of two parties, each
measuring a pair of two-outcome observables. The only
nontrivial tight Bell inequality in this scenario --- the CHSH Bell inequality \cite{CHSH} ---
can be written in the ``probability'' form as
\begin{equation}\label{CHSHprob}
\sum_{a_1,a_2=0}^1\sum_{x_1,x_2=1}^{2}P(a_1\oplus a_2=(x_1-1)(x_2-1)|x_1,x_2)\leq 3,
%
%
\end{equation}
where $\oplus$ stands for addition modulo two.



Let us notice that in the case when all measurements have only two outcomes, i.e., $d=2$,
correlations can be equivalently expressed via expectation values
\begin{equation}
\left\langle \mathcal{M}_{x_{i_1}}^{(i_1)}\ldots \mathcal{M}_{x_{i_k}}^{(i_k)}\right\rangle
\end{equation}
with $x_{i_1},\ldots,x_{i_k}=1,\ldots,m$, $i_1<\ldots<i_k=1,\ldots,n$
and $k=1,\ldots,n$. Here, $\mathcal{M}^{(i)}_{x_i}$ denote observables with outcomes
labelled by $a_i=\pm 1$; in particular, in the quantum-mechanical case
$\mathcal{M}_{x_i}^{(i)}=M_{x_i}^{+1}-M_{x_i}^{-1}$. Both representations
are related to each other through the following set of formulas
%
$$p(a_1,\ldots,a_n|x_1,\ldots,x_n)=\frac{1}{2^n}
\left(1+\sum_{k=1}^n\sum_{1\leq i_1<\ldots <i_k\leq n}a_{i_1}\cdot\ldots\cdot a_{i_k}\langle \mathcal{M}_{x_{i_1}}^{(i_1)}\ldots \mathcal{M}_{x_{i_k}}^{(i_k)}\rangle\right).$$
The advantage about the ``correlator'' picture is that the non-signalling
conditions are already incorporated in it. On the other hand, the correlators must
satisfy a set of inequalities corresponding to the non-negativity conditions of
probabilities $p(a_1,\ldots,a_n|x_1,\ldots,x_n)\geq 0$.

To illustrate the ``correlator representation'', let us
consider again the simplest $(2,2,2)$ scenario. The eight independent conditional
probabilities fully describing correlations in this scenario are equivalent to
eight expectation values $\langle \mathcal{M}_{x_1}^{(1)}\mathcal{M}_{x_2}^{(2)}\rangle$,
$\langle \mathcal{M}_{x_1}^{(1)}\rangle$, and $\langle \mathcal{M}_{x_2}^{(2)}\rangle$
with $x_1,x_2=1,2$. Also, the CHSH Bell inequality (\ref{CHSHprob})
can be rewritten in its ``more standard'' form as
\begin{equation}
|\langle \mathcal{M}_{1}^{(1)}\mathcal{M}_{1}^{(2)} \rangle+\langle \mathcal{M}_{1}^{(1)}\mathcal{M}_{2}^{(2)}\rangle +\langle \mathcal{M}_{2}^{(1)}\mathcal{M}_{1}^{(2)}\rangle-\langle \mathcal{M}_{2}^{(1)}\mathcal{M}_{2}^{(2)} \rangle |\leq 2.
\end{equation}



From now on we concentrate on the $(n,2,2)$ scenario (two two-outcome measurements).
The complexity of characterizing the corresponding classical
polytope is enormous. It is fairly easy to see that the number of its
vertices (extremal points) is equal to 2 to the number of all
possible choice of parties and observables, i.e., $2^{2n}$, so it
grows exponentially with $n$. The dimension of the space
of probabilities is the number of choices of measurements
by each party, which is $2+1$, since each party has at their
disposal $2$ observables or it may not measure anything, to the
power $n$. One then has to subtract 1 from this result, since if all
parties do not measure, the result is trivial. Clearly, the resulting dimension
$3^n-1$ grows exponentially with the number of parties. It is then
not surprising at all that the problem of characterization of the
classical polytope is, depending on formulation, NP-complete or
NP-hard.  Already for few parties finding all Bell inequalities is
an impossible task.
%


\subsection{Detecting non-locality in many body systems with two-body correlators}

Clearly, if we want to find Bell inequalities for many-body
systems, we need some simplifications. This was the idea behind
the recent papers \cite{Science,TIpaper}, which focus on Bell
inequalities involving one- and two-body correlators. In what
follows we will refer to such Bell inequalities as
\textit{two-body Bell inequalities}. Notice in
passing that several criteria allowing for entanglement detection
in many-body systems from such quantities are already
known (see, e.g., Refs. \cite{Toth05,Korbicz05,Toth07,Cramer11,Laskowski13}).

Restricting the study to low-order correlations reduces the
dimension of the space and, thus, may simplify the problem of finding
non-trivial Bell inequalities.
However,
%
%
it is not as simple as it
sounds. First, one wants these Bell inequalities to be valid for any number of
parties which, due to the fact that the complexity of the set of
classical correlations grows exponentially with $n$, usually
appears to be a very difficult task. Second, one wants such Bell
inequalities to be useful, that is, to be capable of revealing
nonlocality in some physically interesting states. However,
intuitively, most of the information about correlations in the
system are contained in high-order correlators, i.e., those
involving many observers, and so Bell inequalities based on them
are expected to be better at detecting nonlocality. All this makes
the task of finding Bell inequalities from two-body correlators
extremely difficult. It should be stressed in passing that as
proven in Refs. \cite{BSV,Wiesniak,Lars} all-partite correlations
are not necessary to detect multipartite nonlocality.



Recently, a positive answer to the above question has been given in Ref. \cite{Science}
by proposing classes of Bell inequalities constructed from one- and two-body
expectation values, and, more
importantly, showing examples of physically relevant many-body
quantum states (i.e. ground states or low energy states of
physically relevant Hamiltonians) violating these inequalities.
Notice that finding and classifying such states is an interesting
task in itself, especially in a view of the fact that many
genuinely entangled quantum many-body states have two-body reduced
density matrices (or in other words covariance matrices) that
correspond to two-body reduced density matrices of some separable
state. This is the case of the so-called {\it graph states}, as
demonstrated in Ref. \cite{Gittsovich}; obviously one cannot
detect entanglement of such states with two-body correlators, not
even mentioning nonlocality.

Let us now briefly describe the way the two-body Bell inequalities were found in \cite{Science}.
First, by neglecting correlators of order larger than two one projects the polytope onto much smaller one $\mathbbm{P}_2$ spanned by two-body and one-body correlations functions.
%
%
    In this way we have achieved a severe reduction of the dimension of the polytope: $\mathrm{dim}\, \mathbbm{P}=3^n-1 \longrightarrow \mathrm{dim}\, \mathbbm{P}_2=2n^2$. Now, the general form of Bell inequalities is
\begin{eqnarray} \label{Raimat}
&\sum_{i=1}^{n}(\alpha_i\langle \mathcal{M}^{(i)}_0\rangle +\beta_i\langle\mathcal{M}^{(i)}_1\rangle)+\sum_{i<j}^{n}\gamma_{ij}
\langle \mathcal{M}_{0}^{(i)}\mathcal{M}_{0}^{(j)}\rangle+\sum_{i\neq j}^{n}\delta_{ij} \langle \mathcal{M}_{0}^{(i)}\mathcal{M}_{1}^{(j)}\rangle\nonumber \\
 &\sum_{i<j}^{n}\varepsilon_{ij} \langle \mathcal{M}_{1}^{(i)}\mathcal{M}_{1}^{(j)}\rangle+ \beta_C\geq 0,
\end{eqnarray}
where for convenience we wrote them down using expectation values instead of probabilities; recall that in the case of all observables having two outcomes both representations are equivalent.

For "interesting" $n$ the inequalities (\ref{Raimat}) still contain too many coefficients; in fact, the dimension of the corresponding polytope grows quadratically with $n$ (for, say, $n=100$, $\dim\mathbbm{P}_2=20000$, i.e., it is still too large). To further simplify the problem one can demand that the Bell inequalities under study obey some symmetries. In particular, in Refs. \cite{Science} and \cite{TIpaper} Bell inequalities obeying
permutational and translational invariance have been considered. In what follows we discuss in more detail the results of Ref. \cite{Science}.

 \subsection{Permutational Invariance}

Let us now restrict our attention to two-body Bell inequalities
that are invariant under a permutation of any two parties. It is
fairly easy to see that their general form reads
\begin{equation}
I:=\alpha \mathcal{S}_0+\beta \mathcal{S}_1+\frac{\gamma}{2}\mathcal{S}_{00}+\delta\mathcal{S}_{01} +\frac{\varepsilon}{2}\mathcal{S}_{11}\geq -\beta_{C}
\end{equation}
    where
\begin{equation}
\mathcal{S}_k=\sum_{i=1}^{n}\langle \mathcal{M}_{k}^{(i)}\rangle, \qquad \mathcal{S}_{kl}=\sum_{i\neq j=1}^{n}\langle \mathcal{M}_{k}^{(i)}\mathcal{M}_{l}^{(j)}\rangle,
\end{equation}
with $k,l=0,1$ are the symmetrized one- and two-body expectation values, respectively.
Geometrically, we have mapped the two-body polytope $\mathbbm{P}_2$
to a simpler one $\mathbbm{P}_2^S$ whose elements are five-tuples $(\mathcal{S}_0,\mathcal{S}_1,\mathcal{S}_{00},\mathcal{S}_{01},\mathcal{S}_{11})$ consisting of the symmetrized expectation values.
Obviously, by doing this projection $\mathbbm{P}_2 \longrightarrow\mathbbm{P}_2^S$
one is able to limit the dimension of the local polytope to 5 and, more importantly,
this number is independent of the number of parties.
Still, the number of vertices of the projected polytope is $2(n^2+1)$, i.e.,
it scales quadratically with $n$, so the characterization of all permutationally invariant two-body Bell inequalities is not trivial at all.

Nevertheless, the following three-parameter classes of
permutationally invariant two-body Bell inequalities was found in
Ref. \cite{Science}
%
%
\begin{equation}\label{Rioja}
%
x[\sigma\mu\pm(x+y)] \, \mathcal{S}_0+\mu y \, \mathcal{S}_1+\frac{x^2}{2} \, \mathcal{S}_{00}+\sigma xy \,\mathcal{S}_{01} +\frac{y^2}{2} \, \mathcal{S}_{11}\geq -\beta_{C},
\end{equation}
where $x,y \in \mathbbm{N}$, $\sigma = \pm 1$
and $\mu \in \mathbbm{Z}$ with opposite parity to $\varepsilon$ ($\gamma$) for odd (even) $n$.
The classical bound is in this case $\beta_{C}^{(\pm)} = (1/2)[n(x+y)^2+(\sigma\mu\pm x)^2-1],$
and it grows linearly with $n$.

\subsection{Symmetric two-body Bell inequalities: example}

Here we consider an exemplary Bell inequality belonging to the class (\ref{Rioja}):
\begin{equation}\label{Murcia}
-2\mathcal{S}_0+\frac{1}{2}\mathcal{S}_{00}-\mathcal{S}_{01}+\frac{1}{2}\mathcal{S}_{11}+2n\geq 0,
\end{equation}
where we have substituted $x=y=-\sigma=1$ and $\mu=0$.

Now, to see whether this Bell inequality is violated by some
quantum states, let us assume that all parties measure the same
pair of observables $\mathcal{M}_0^{(i)} =  \mathcal{M}_0 =
\sigma_z$ and $\mathcal{M}_1^{(i)} =  \mathcal{M}_1 = \cos\theta
\sigma_z + \sin\theta \sigma_x$ with $i=1,\ldots,N$. Fig.
\ref{fig:Sci1} shows two plots of the ratio $Q_v/\beta_C$ of the maximal
quantum violation $Q_v$ of this inequality with the above settings and
the classical bound $\beta_C=2n$. In the left plot we show the dependence on
$n$. The relative violation remains significant (of order 1) for
$n$ of order of $10^4$, and seems to grow or to saturate at large
$n$. In the right plot we show maximal violation as function of
the angle $\theta$ that defines the second observable. Again, the
maximal violation remains significant for a large set of angles
close to the optimal one.
\begin{figure}
\centering    \includegraphics[width=0.45\columnwidth]{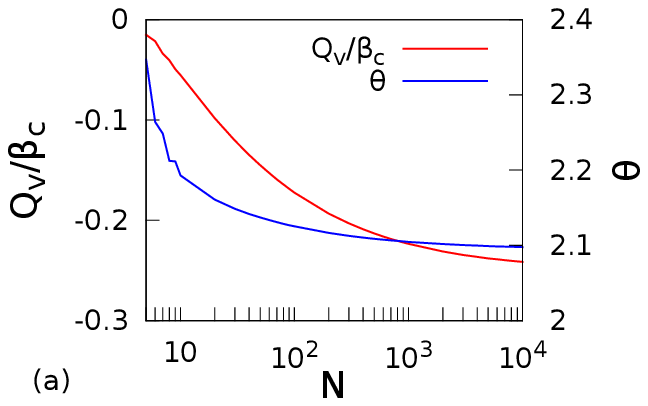}
    \includegraphics[width=0.35\columnwidth]{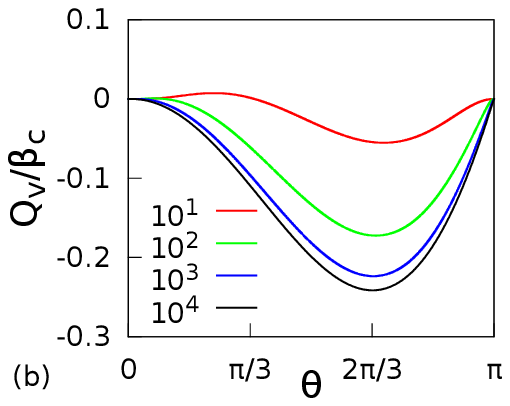}
\label{fig:Sci1}
\caption{Plot (a) presents the ratio $Q_v/\beta_C$ of the maximal violation of the inequality (\ref{Murcia}) and its classical bound $\beta_C=2n$ along with $\theta$ realizing this maximal violation, both as function of the number of particles $n$. Plob (b) shows the robustness of the violation under misalignments of the second measurement defined by the angle $\theta$ for four values of $n$. }
    \end{figure}
%

  \subsection{Many-body symmetric states}

The next question to answer is what are the state that violated the two-body Bell inequalities, and which states can be detected by measuring these inequalities? To this aim we considered the Lipkin-Meshkov-Glick Hamiltonian \cite{LMG}, which is commonly used in nuclear  physics, and more recently in trapped atoms and trapped ions physics
    \begin{equation}
    H=-\frac{\lambda}{n}\sum_{i,j=1, i<j}^{n}\left(\sigma_x^{(i)} \sigma_x^ { (j) } + \sigma_y^{(i)}\sigma_y^{(j)}\right)-h\sum_{i=1}^{n}\sigma_z^{(i)}.
    \end{equation}
    Its ground state is the famous Dicke state \cite{Dicke}: for $n$ even it is a symetric combination of all states with exaclty $n/2$ zeros (spins down),
    $|{D_n^{n/2}}\rangle=\mathcal{S}(|{\{0,n/2\},\{1,n/2\}}\rangle)$
     But , for $n$ odd, we have a doubly degenerate ground state
    $|D_n^{[ n/2]}\rangle$ or  $|D_n^{[n/2]+1}\rangle$, for which integer part of $n/2$ or integer part of $n/2$ plus one spin are down.

It was shown in Ref. \cite{Science} that the nonlocality of these states can be revealed with the aid of the
following Bell inequalities
%
%
\begin{equation}\label{DickeBell}
%
%
\alpha_n \mathcal{S}_0+\beta_n \mathcal{S}_1+\frac{\gamma_n}{2}\mathcal{S}_{00}+\delta_n\mathcal{S}_{01} +\frac{\varepsilon_n}{2}\mathcal{S}_{11}\geq -\beta_{C}^n,
\end{equation}
    with $\alpha_n = n(n-1)(\lceil n/2\rceil-n/2) = n\, \beta_n$, $\gamma_n = n(n-1)/2$, $\delta_n = n/2$,  $\varepsilon_n = -1, $ and the classical bound is found to be
    $\beta_C^n=(1/2)n(n-1)\left\lceil(n+2)/2\right\rceil$.
    Again the observables are taken as
    $\mathcal{M}_0^{(i)} =  \mathcal{M}_0 = \sigma_z$, $\mathcal{M}_1^{(i)} =  \mathcal{M}_1 = \cos\theta \sigma_z + \sin\theta \sigma_x.$
    The results are presented in Fig. \ref{Lipkinviol}. We see that (\ref{DickeBell}) is violated by the ground state of the LMG Hamiltonian, and that the violation is not so large, but significant; this time it actually decreases slowly with $n$.
    \begin{figure}
    \centering
    \includegraphics[width=0.3\columnwidth]{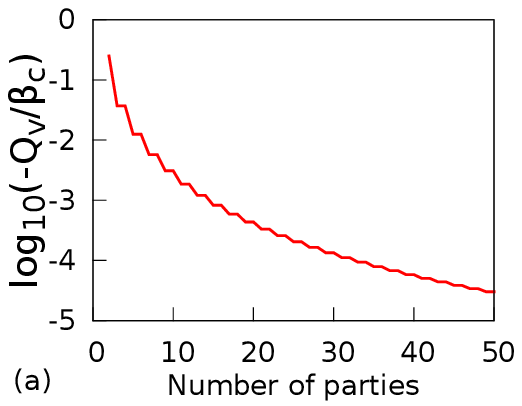}
    \hskip 1cm
    \includegraphics[width=0.3\columnwidth]{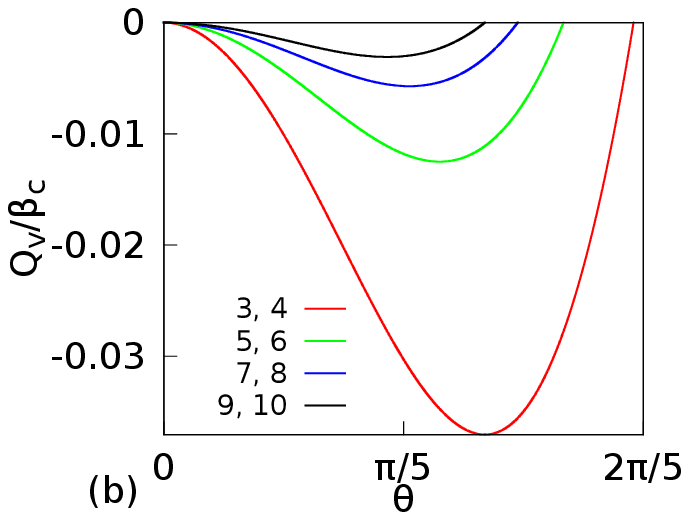}
    \label{Lipkinviol}
    \caption{Plot (a) presents the violation of the Bell inequality (\ref{DickeBell}) by the Dicke states as a function of $n$. Plot (b) presents the violation as a function of the angle defining the second measurement $\mathcal{M}_1$.}
    \end{figure}

At this point it is worth mentioning that the detection of
nonlocality in this case can be realized by measuring the total spin
components and their fluctuations: these quantities can be
measured with a great precision in current experiments with cold
atoms and ions, for instance using the spin polarization
spectroscopy. Indeed, the considered Bell inequality requires
measurements of $\mathcal{S}_0 = 2\langle S_z\rangle$,
$\mathcal{S}_1 = 2\langle{\bf m}\cdot{\bf S}\rangle$,
$\mathcal{S}_{00} = 4\langle S_z^2\rangle-n$, $\mathcal{S}_{11} =
4\langle({\bf m}\cdot {\bf S})^2\rangle - n$, and $\mathcal{S}_{01} =
(1/4)[\langle(S_z+{\bf m}\cdot{\bf S})^2\rangle-\langle
(S_z-{\bf m}\cdot{\bf S})^2\rangle$, where ${\bf m}$ is a unit vector
determining the spin direction in the second measurement.

It is worth mentioning that in the second paper \cite{TIpaper}, the
more complex case of translational invariance was also considered
-- to this aim the parties were enumerated as if they were located
in a 1D chain with periodic boundary conditions (a 1D ring). Recall that in this case the
general form of a Bell inequality is
\begin{equation}
\alpha \mathcal{S}_0+\beta \mathcal{S}_1+\sum_{k=1}^{\lfloor n/2\rfloor}\left(\gamma_k\mathcal{T}_{00}^{(k)}+\epsilon_k\mathcal{T}_{11}^{(k)}\right)+
\sum_{k=1}^{N-1}\omega_k\mathcal{T}_{01}^{(k)}\geq -\beta_C,
\end{equation}
where $\mathcal{S}_i$ $(i=0,1)$ are defined as before and $\mathcal{T}$'s are translationally invariant two-body correlators given by
\begin{equation}
\mathcal{T}_{ij}^{(k)}=\sum_{m=1}^{n}\left\langle \mathcal{M}_i^{(m)}\mathcal{M}_j^{(m+k)}\right\rangle\qquad (i\leq j=0,1)
\end{equation}
with $k=1,\ldots,\lfloor n/2\rfloor$ for $i=j$ and $k=1,\ldots,n-1$ for $i< j$.
The number of coefficients is now of order of 3$n$, so the problem becomes
intractable for $n$ large. We have, nevertheless found and
classified all tight Bell inequalities for 3 and 4 parties, and
provided some examples of five-party Bell inequalities that involve
correlators between next neighbours only.

\section{Conclusions}
\label{Conclusions}

Let us conclude by listing several experimental setups in which
nonlocality in many-body systems may be tested using the two-body
Bell inequalities:

\begin{itemize}

\item {\bf Ultracold trapped atoms.} Dicke states have been recently created in spinor Bose-Einstein condensates
(BEC) of Rubidium $F=1$ atoms, via the parametric process of spin
changing collisions, in which two $m_F=0$ atoms collide to produce a
$m_F=\pm 1$ pair  \cite{Klempt}.   These recent experiments
demonstrate the production of many thousands of neutral atoms
entangled in their spin degrees of freedom. Dicke-like states can
be produced in this way, with at least 28-particle genuine
multi-party entanglement and  a generalized squeezing parameter of
$−11.4(5)$ dB. Similarly, Rubidium atoms of pseudo-spin 1/2 in
BEC may be employed to generate scalable squeezed states (for the
early theory proposal see \cite{spin-sq}, for experiments
see\cite{Muessel}).  Very recently non-squeezed (non-Gaussian)
entangled states of many atoms \cite{Strobel} were generated in
this way. The number of atoms used in these experiments are of
order of thousands and larger. So far, these numbers and
experimental errors and imperfections are too large, while the
corresponding fidelities too small to detect many body
non-locality. In principle, however, it is possible to perform
these experiments with mesoscopic atom numbers (say $\le 100$),
controlling the atom number to a single atom level (see Ref.
\cite{Hume} for the resonant fluorescence detection of Rb$^{87}$
atoms in a MOT, and Refs. \cite{Wenz,Zuern} for optically trapped
spin 1/2 fermions).

\item{\bf Ultracold trapped ions.} Ultracold trapped ions with internal
pseudo-spin "talk" to each other via phonon excitations, and in
some condition behave as spin chains with long range interaction.
This was originally proposed in Ref. \cite{Wunderlich}, using
inhomogeneous magnetic fields, and in Ref. \cite{Porras},
employing appropriately designed laser-ion interactions. The
pioneering experiments were reported in Refs.
\cite{Schaetz,Monroe2010}. While in the early theory studies
\cite{Porras,Porras2,HaukeNJP,Maik}  spin interactions decaying
with the third power of the distance were considered, it was experimentally
demonstrated  that  management of phonon dispersion allows to
achieve powers between 0.1 and 3  in the 2D arrays of
traps \cite{Bollinger}. Recent state of art represents the  work
on experimental realization of a quantum integer-spin chain with
controllable interactions \cite{Monroe}. We have studied trapped
ion systems in relation to long range $SU(3)$ spin chains and
quantum chaos \cite{Grass}, and trapped-ion quantum simulation of
tunable-range Heisenberg chains \cite{Grass1}. In the latter paper
we demonstrated that significant violation of the Bell
inequalities, discussed in this lecture, is possible even for
the ground states of the models with large, but finite interaction
range. The experimental scheme is presented in Fig. \ref{toby}.

\begin{figure}[h]
\centering
\includegraphics[width=0.5\textwidth]{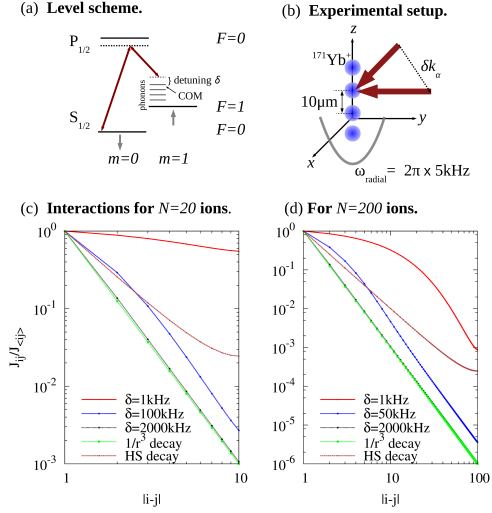}
\caption{\label{toby}
(a,b): Level scheme and setup for a possible implementation of spin-spin interactions in $^{171}$Yb$^+$.
(c,d):
Interaction strengths $J_{ij}$ between one ion in the center and the other
ions, for (c) $N=20$ or (d) $N=200$, and different detunings from
the center-of-mass (COM) mode. We have used the parameters specified in (b). Interactions are compared with the interactions of the Haldane-Shastry (HS) model (brown lines), and with $1/r^3$ interactions (green lines).
}
\end {figure}

\item {\bf Ultracold atoms in nanostructures.} Yet another possibility concerns systems of
ultracold atoms trapped in the vicinity of tapered fibers and optical
crystals (band gap materials). The experimental progress in
coupling of ultracold atomic gases to nanophotonic waveguides,
initiated by Refs. \cite{Nayak,Vetsch,Goban}, is very rapid (cf.
\cite{Atomlight}). Early theoretical studies concerned remarkable
optical properties of these systems (cf.
\cite{Kien,Zoubi,Gong,Gorshkov}).  Ideas  and proposals concerning
realization of long range spin models were developed more
recently, and mainly  in Refs. see \cite{Chang, Chang1, Chang2}.

\item{\bf Cold and ultracold atomic ensembles.} Last, but not least, one should consider cold and ultracold ensembles (for an excellent  review
 see \cite{Hammerer}), in which, by employing quantum Faraday effect, one can reach unprecendented
degrees of squeezing of the total atomic spin (cf.
\cite{Napolitano,Sewell}), and unprecendented degrees of  precision
of quantum magnetometry (cf. \cite{magneto}). Note that in many
concrete realisations the many body Bell inequalities derived in
this paper require precise measurements of the total spin
components, and their quantum fluctuations. Quantum Faraday
effect, or in other words spin polarization spectroscopy, seems to
be  a perfect method to achieve this goal; note that in principle
 it allows also to reach spatial resolution, and/or  to measure spatial Fourier components of the total spin \cite{Isart,Rogers}.

\end{itemize}

\acknowledgments We thank Joe Eberly, Christian Gogolin, Markus
Oberthaler, Julia Stasi\'nska, and Tam\'as V\'ertesi for
enlightening discussions. This work was supported by Spanish
MINECO projects FOQUS, DIQIP CHIST-ERA, and AP2009-1174 FPU PhD
grant. We acknowledge also EU IP SIQS, ERC AdG OSYRIS, ERC CoG
QITBOX, and The John Templeton Foundation. R. A. acknowledges
Spanish MINECO for the Juan de la Cierva scholarship.

\end{document}